\documentclass[journal,10pt]{IEEEtran} 
\usepackage[utf8]{inputenc}
\usepackage{xcolor}
\usepackage{amsmath,amsfonts,amssymb,amsthm}
\usepackage{mathtools}
\usepackage{cite}
\usepackage{tikz}
    \usetikzlibrary{shapes}
\usepackage{tikz-3dplot}
\usepackage{hyperref}

\def\reals{\mathbb{R}}
\def\nnreals{\reals_+}

\def\QuanSet{\mathcal{Z}}
\newcommand{\ToPars}[3]{\operatorname{\mathcal{W}}_{#1}\!\left(#2,#3\right)}
\def\setA{\mathcal{A}}
\newcommand{\PD}[2]{\mathcal{M}^{+}_{#1}\left(#2\right)}
\newcommand{\PDreals}[1]{\PD{#1}{\reals}}
\newcommand{\PDbounded}[1]{\mathcal{M}^{\oplus}_{#1}\left(\reals\right)}
\newcommand{\matrices}[2]{\mathcal{M}_{#1}\left(#2\right)}
\newcommand{\matricesreals}[1]{\matrices{#1}{\reals}}
\newcommand{\diag}[2]{\mathcal{D}_{#1}\left(#2\right)}

\newcommand{\diagPD}[1]{\mathcal{D}_{#1}^+}

\def\locpar{\mathbf{x}}
\def\scapar{\boldsymbol{\Psi}}
\def\scalarscapar{\psi}
\def\diagscapar{\boldsymbol{\Lambda}}

\def\noise{\mathbf{w}}
\def\contdata{\mathbf{y}}
\def\obsmat{\mathbf{S}}
\def\matrixC{\mathbf{C}}
\newcommand{\id}[1]{\mathbf{I}}
\newcommand{\zero}[1]{\mathbf{0}_{#1}}

\newcommand{\pdf}[2]{\operatorname{f_{#1}}\!\left(#2\right)}
\newcommand{\likelihood}[2]{\operatorname{\mathcal{L}}\!\left(#1;#2\right)}
\newcommand{\prob}[1]{\operatorname{P}_{#1}}

\newcommand{\normal}[2]{\operatorname{\mathcal{N}}\!\left(#1,#2\right)}

\def\quantizer{\operatorname{Q}}
\def\invquantizer{\operatorname{Q}^{-1}}

\newtheorem{theorem}{Theorem}
\newtheorem{lemma}[theorem]{Lemma}
\newtheorem{proposition}[theorem]{Proposition}

\theoremstyle{remark}
\newtheorem{remark}[theorem]{Remark}

\title{Convex Quantization Preserves Logconcavity}

\author{Pol~{del Aguila Pla},~\IEEEmembership{Member,~IEEE,}\thanks{
        Pol~{del Aguila Pla} is with the CIBM Center for Biomedical Imaging, in
        Switzerland.} Aleix~Boquet-Pujadas,\thanks{
        Pol~{del Aguila Pla} and Aleix~Boquet-Pujadas are with the Biomedical Imaging Group
        at the École polytechnique fédérale de Lausanne, in Lausanne,
        Switzerland.} Joakim~Jald\'{e}n,~\IEEEmembership{Senior~Member,~IEEE}\thanks{Joakim~Jald\'{e}n is with the Division of Information Science and Engineering at the School of Electrical Engineering and Computer Science at the KTH Royal Institute of Technology, in Stockholm, Sweden.}}
\begin{document}

\maketitle

    \begin{abstract}
        A logconcave likelihood is as important to proper statistical inference 
as a convex cost function is important to variational optimization.
Quantization is often disregarded when writing likelihood models, 
ignoring the limitations of the physical detectors used to collect the data. 
These two facts call for the question: would including quantization in likelihood models preclude logconcavity? are the true data likelihoods 
logconcave? We provide a general proof that the same simple 
assumption that leads to logconcave continuous-data likelihoods also
leads to logconcave quantized-data likelihoods, provided that convex
quantization regions are used.

    \end{abstract}
    \begin{IEEEkeywords}
        Bayesian statistics, Likelihood, Privacy-aware data analysis, 
        1-bit compressed sensing, Inverse problems.
    \end{IEEEkeywords}

	\section{Introduction}
	    Inference from signals in the digital domain is of central          %
importance in digital signal processing. However, the discrete 
nature of measurement devices is often disregarded or misrepresented
in building data likelihood models. When quantization is deemed fine
enough, the established procedure is to ignore it or to model it as
additive noise \cite{Widrow2008,Azizzadeh2019}.	Conversely,
the few works investigating coarse quantization do so under 
simplifying assumptions \cite{Li2017,Ren2019,Khobahi2019},
suggesting that optimal estimation might be intractable.
In this paper, we show that the exact likelihood of quantized data 
remains logconcave under widely applicable assumptions.

This result is of interest to several signal processing domains. 
For example, coarse quantization has been increasingly popular
 due to the 
advent of 1-bit compressed sensing
 techniques~\cite{Wen2016,Zayyani2016,Zayyani2016a,Li2017,Stein2018,Ren2019,Khobahi2019}, which can be seen as its limiting case. These promise to 
incorporate low-cost high-speed analog-to-digital converters (ADC)
into wireless-communications pipelines.

Quantization is also influential in the privacy-enhancement         %
literature~\cite{Aggarwal2008,Gao2018}. Besides simple 
noise-addition mechanisms, coarse quantization or ``aggregation'' is
one of the most straightforward techniques to induce differential 
privacy and $k$-anonymity on a database. However, privacy protection 
and data utility are conflicting objectives. To better understand 
this tradeoff in aggregation-based techniques, it is fundamental to 
characterize the properties of the likelihood after 
quantization.

Variational inverse problems with convex data terms, which can be   %
interpreted as maximum-a-posteriori (MAP) estimators for logconcave
likelihood models, also typically ignore data quantization in their
cost functions. However, examples where quantization plays a crucial
role abound in applications (see Figure~\ref{fig:regions}): think of 
honeycomb electrode detectors~\cite{Zhang2018}~in high-energy 
physics, the sinogram-domain quantization induced by the geometry of 
positron emission tomography (PET) scanners, the pixelized detectors 
in computed tomography~\cite{Danielsson2021}, or CCD and CMOS 
photodetectors in biological microscopy~\cite{Koklu2012}. 
Furthermore, we expect that Bayesian image reconstruction algorithms  will benefit even further from 
treating quantization explicitly because they use the entire 
posterior distribution to make inferences. This is in contrast to 
using only the mean or the mode, which are much more easily 
preserved between continuous and quantized data.

Admittedly, incorporating quantization into models without
some guarantee of likelihood logconcavity precludes many
theoretical guarantees.
Examples include tractable maximum likelihood (ML) and MAP
estimates, connected and convex credible and confidence sets, as well as large-sample
normal approximations that lead to usable hypothesis
testing~\cite{Pratt1981}. 
These are automatic for most common noise
models (e.g., Gaussian or exponential) in the continuous-data setting.
To recover these properties, most existing works studying quantization
resort to simplifying assumptions such as
Gaussian data and 1-bit quantization \cite{Li2017,Ren2019,Khobahi2019}, or
optimistically assume logconcavity \cite{Wen2016,Gao2018,Stein2018}.
In contrast, we prove that the logconcavity of the quantized-data likelihood follows from
the same assumption made on its continuous-data 
counterpart: logconcavity of the noise distribution. The result holds for
convex quantizers, which include not only all the examples mentioned so
far, but also an overwhelming majority of applications. This means that the corresponding likelihood-based methods automatically enjoy convergence guarantees. Therefore, more accurate models can be readily adopted for a plethora of detector models. We expect that the convenience of this result will
encourage researchers to more often include quantized models and thus
better capture the nature of measurement devices. Similarly, another
aim of this letter is to publicize and exploit a series of techniques that
are part of the folklore of the statistics literature but are yet to be made known
to the signal processing community~\cite{Burridge1982}. For
example, while not entirely new to our
community~\cite{Boyd2004,Conti2009,Msechu2012},  Pr\'{e}kopa's results
\cite{Prekopa1973} on logconcavity are key to this work and have not previously
been harnessed in their full generality.

To the best of our knowledge, our analysis is the most general treatment of
likelihood logconcavity for quantized data yet. Our study also accounts for
the scale parameter, which has only been considered before in a
specific application in~\cite{Ren2019}. Furthermore, we appear to be
the first to do the analysis for a generic vector quantizer $Q$, which is more general
than simple combinations of independent ADCs.

In Section~\ref{sec:cont}, we introduce the topic of likelihood logconcavity
by studying it for continuous data models. In Section~\ref{sec:proposed}, we present
our main result: quantizers with convex quantization regions (convex quantizers) yield
logconcave likelihoods when the underlying noise model has a logconcave 
pdf. In Section~\ref{sec:proofs} and the Appendix, we provide the proof for this
statement, as well as supporting mathematical results.


    \section{Logconcavity for Continuous Data}
\label{sec:cont}

    Our central claim is that for both continuous and quantized data,
    likelihood logconcavity follows from the logconcavity of the noise pdf.
    This result is known for continuous data. We present it in
    Proposition~\ref{th:cont}. 

    We start by introducing the context and notation. Consider a data vector $\contdata\in\reals^n$ modeled as
    \begin{IEEEeqnarray}{c}
        \contdata = \scapar^{-1} \left( \obsmat \locpar + \noise \right)\,,
        \label{eq:LM}
    \end{IEEEeqnarray}
    where $\scapar\in\PDreals{n}$, $\obsmat\in\matricesreals{n,m}$, and $\locpar\in\reals^m$,
    with $\PDreals{n}$ representing the set of $n\times n$ real positive-definite matrices and
    $\matricesreals{n,m}$ the space of $n\times m$ real matrices. Additionally,
    $\noise\in\reals^n$ is a random noise vector $\noise\sim\pdf{\noise}{\noise}$ drawn from the 
    logconcave pdf $\pdf{\noise}{\cdot}$, \emph{i.e.}, for any $\alpha \in [0,1]$ and any two 
    $\noise_0,\noise_1\in\reals^n$,
    \begin{IEEEeqnarray}{rl} \label{eq:logconcavity_pdf}
        \pdf{\noise}{\noise_\alpha} \geq
        \pdf{\noise}{\noise_1}^\alpha \pdf{\noise}{\noise_0}^{(1-\alpha)}\!\!.
    \end{IEEEeqnarray}
    Here, $\noise_\alpha \coloneqq \alpha \noise_1 + (1-\alpha) \noise_0$ is a convex combination between $\noise_0$ and $\noise_1$.
    Throughout the paper, we will resort to this notation for convex
    combinations for simplicity of exposition. The case $\scapar = \id{n}$ in
    \eqref{eq:LM} corresponds to a usual linear model formulation, where $\obsmat$ is the observation matrix. In
    statistics, $\locpar$ and $\scapar$ in \eqref{eq:LM} are known as the
    location and scale parameters of the distribution family defined by
    \eqref{eq:LM}. For example, if the noise comes from a standard
    multivariate normal distribution $\noise\sim\normal{\mathbf{0}}{\id{n}}$, we have
    that $\contdata \sim \normal{\scapar^{-1}\obsmat\locpar}{\scapar^{-2}}$.
    
    \begin{remark}[Parametrization of Location-Scale Statistical Models]
        Although the location-scale parametrization in \eqref{eq:LM} is less
        familiar than, say, the mean and covariance matrices in multivariate
        Gaussian models, this choice is essentially related to
        the study of logconcavity. In fact, to show that joint estimators of $\boldsymbol \mu$ and $\boldsymbol \Sigma$ in a
        $\normal{\boldsymbol\mu}{\boldsymbol\Sigma}$ model are tractable, one needs to first
        reparametrize the problem in terms of $\locpar$ and $\scapar$.
    \end{remark}
    
    For each vector $\contdata\in\reals^n$, the likelihood
    is a function of $\locpar$ and $\scapar$ such that
    $\likelihood{\cdot,\cdot}{\contdata}\colon\reals^m\times\PDreals{n}
    \to\nnreals$ and
    \begin{equation} \label{eq:likelihood}
        \likelihood{\locpar,\scapar}{\contdata} \coloneqq
        \pdf{\contdata;\locpar,\scapar}{\contdata} = 
        \pdf{\noise}{\scapar \contdata - \obsmat \locpar}.
    \end{equation}
    Here, we used \eqref{eq:LM} to write $\likelihood{\locpar,\scapar}{\contdata}$ in terms of  $\pdf{\noise}{\cdot}$, and
    $ \pdf{\contdata;\locpar,\scapar}{\contdata}$ denotes the pdf
    	 of $\contdata$ for some given values of $\locpar$ and $\scapar$.
    The following result sets the stage for our study.

    \begin{proposition}[Logconcave Noise Generates Logconcave Likelihoods]
    \label{th:cont}
        Consider a sample $\contdata\in\reals^n$ drawn from model
        \eqref{eq:LM} and assume that the pdf of the noise,
        $\pdf{\noise}{\cdot}$, is logconcave. Then,
        $\likelihood{\locpar,\scapar}{\contdata}$ is jointly logconcave in
        $\locpar$ and $\scapar$.
    \end{proposition}
    \begin{proof}
        Using \eqref{eq:likelihood}, we obtain
        \begin{IEEEeqnarray*}{rl}
            \likelihood{\locpar_\alpha,\scapar_\alpha}{\contdata} &=
            \pdf{\noise}{\scapar_\alpha \contdata - \obsmat \locpar_\alpha}\\
            &\geq \pdf{\noise}{
                    \left[\scapar_1\contdata - \obsmat\locpar_1\right]
                              }^\alpha
            \pdf{\noise}{
                    \left[\scapar_0 \contdata - \obsmat\locpar_0\right]
                        }^{1-\alpha},
        \end{IEEEeqnarray*}
        which is the desired result. Here, we used \eqref{eq:logconcavity_pdf}
        with
            $\noise_i = \scapar_{i} \contdata - \obsmat \locpar_{i}$ for 
            $i\in\lbrace0,1\rbrace$.
    \end{proof}

\section{Logconcavity for Quantized Data}\label{sec:proposed}
	It turns out that a statement similar to Proposition~\ref{th:cont} can be made
    for quantized observations $z$ modeled as
    \begin{IEEEeqnarray}{c} \label{eq:QLM}
        z = \quantizer\left(\contdata\right) =
            \quantizer\left(
                            \scapar^{-1}\left(\obsmat\locpar+\noise\right)
                      \right)\!.
    \end{IEEEeqnarray}
	Here, we take the most general view of quantization: we define a quantizer
	as a mapping $\quantizer\colon\reals^n\to\QuanSet$, where $\QuanSet$ is
    a countable set.
    Such a quantizer does not generally treat each dimension of $\contdata$
    independently. The results we
    present hereafter apply to the subclass of quantizers that have convex quantization regions. We call these \emph{convex quantizers} $\quantizer$, and they fulfill that $\invquantizer(z)$ is a convex set $\forall z\in\QuanSet$
    (see Figure~\ref{fig:regions} for examples).
    Among others introduced beforehand, these include
    quantizers composed of independent (monotonic) ADCs for each dimension. Indeed, in this case, for any $z\in\QuanSet$ and all $j\in\lbrace 1,2,\dots,n\rbrace$, there are  $a_j(z), b_j(z)\in\overline{\reals}$ such that $\invquantizer(z) = \prod_{j=1}^n [a_j(z),b_j(z)]$, which is trivially convex. Here, $\overline{\reals}=\reals \cup \lbrace -\infty, \infty\rbrace$.
    Our main result below parallels Proposition~\ref{th:cont} for quantized data obtained from convex quantizers.

    \begin{theorem}[Logconcave Noise and Convex Quantizers Generate Logconcave Likelihoods]\label{th:result}
        Consider a sample $z\in\QuanSet$ drawn from model \eqref{eq:QLM}. Assume that $\quantizer$ is a convex quantizer and that the pdf $\pdf{\noise}{\cdot}$ of the noise is logconcave. Then,
        \renewcommand{\theenumi}{\alph{enumi}}
        \begin{enumerate}
            \item for a given scale parameter $\scapar_0\in\PDreals{n}$, the likelihood $\likelihood{\locpar,\scapar_0}{z}$ is
                logconcave with respect to $\locpar$,
            \item for scale parameters of the form $\scapar=\scalarscapar \id{n}$ with
                $\scalarscapar>0$, the likelihood $\likelihood{\locpar,\scalarscapar \id{n}}{z}$
                is jointly logconcave with respect to $\locpar$ and $\scalarscapar$,
            \item for diagonal positive-definite scale parameters $\scapar=\diagscapar$, the likelihood $\likelihood{\locpar,\diagscapar}{z}$ is jointly logconcave with respect to
                $\locpar$ and $\diagscapar$ if $\quantizer$ is composed of independent ADCs for each dimension.
        \end{enumerate}
        \renewcommand{\theenumi}{enumi}
    \end{theorem}

	The rest of this letter is dedicated to proving Theorem~\ref{th:result}. 
		The proof can be found at the end of the letter.

    \begin{figure}
        \centering
        \input{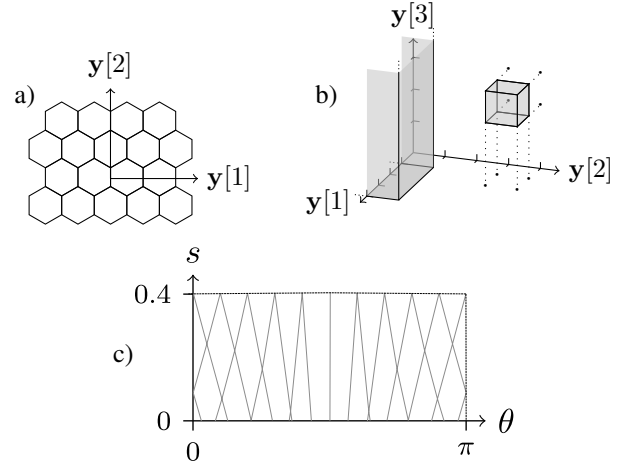}

        \vspace{-10pt}
        
        \caption{a) Convex regions laid down according to a honeycomb pattern, as done in hexagonal electrode detectors. b) Two quantization regions of a convex quantizer formed by independent ADCs. c) Example quantization regions of a PET scanner in the sinogram domain.
            \label{fig:regions}}
    \end{figure}

    \section{Mathematical Study of Logconcavity Under Quantization}
    \label{sec:proofs}
            We start by constructing the likelihood of the location and scale parameters for a given $z\in\QuanSet$. By definition, observing $z$ under model \eqref{eq:QLM} implies that $\contdata\in \invquantizer\left(z\right)$. In turn, this implies that the random noise $\noise$ is within a specific region. For specific values of $\locpar$ and $\scapar$, and for each $z\in\QuanSet$, we define
    \begin{IEEEeqnarray}{c} \label{eq:topars}
        \ToPars{z}{\locpar}{\scapar} \coloneqq \left\lbrace \noise\in\reals^n :
        \contdata \in\invquantizer\left(z\right)\right\rbrace
    \end{IEEEeqnarray}
    as a shorthand for this region. We then write the likelihood as
    \begin{IEEEeqnarray}{c}\label{eq:quantized_likelihood}
        \likelihood{\locpar,\scapar}{z} =
        \prob{\noise}\left[ \ToPars{z}{\locpar}{\scapar} \right],
    \end{IEEEeqnarray}
    where $\prob{\noise}$ is the probability law in $\reals^n$ given
    by the pdf $\pdf{\noise}{\cdot}$.

    Our strategy to prove Theorem~\ref{th:result} combines this compact expression of the quantized-data likelihood with a key result by Pr\'{e}kopa, which we include here for completeness.

    \begin{theorem}[Pr\'{e}kopa's Theorem {\cite[p.~2, Theorem~2]{Prekopa1973}}]
    \label{th:prekopa}
        Let $\noise$ be a continuous random vector in $\reals^n$ with
        logconcave pdf $\pdf{\noise}{\cdot}$ in the sense of \eqref{eq:logconcavity_pdf}.
        Let $\prob{\noise}:2^{\reals^n}\rightarrow[0,1]$
        be the probability measure induced by $\noise$ on $\reals^n$, where
        $2^{\reals^n}$ denotes the set of all possible sets in $\reals^n$.
        Then, for any two
        convex sets $\setA_0,\setA_1\subseteq \reals^n$
        we have that
        \begin{IEEEeqnarray}{c} \label{eq:prekopa}
            \prob{\noise}\left[\setA_\alpha\right] \geq
            \prob{\noise}\left[\setA_1\right]^\alpha
            \prob{\noise}\left[\setA_0\right]^{(1-\alpha)},
        \end{IEEEeqnarray}
        where $\setA_\alpha$ is the Minkowski sum $\alpha \setA_1 + (1-\alpha)\setA_0$.
    \end{theorem}
    Here, the weighted Minkowski sum $\setA_\alpha$ (illustrated in Figure~\ref{fig:minkowski} for $\alpha=1/2$) is the set of all possible
    combinations $\noise_\alpha = \alpha \noise_1+(1-\alpha)\noise_0$ in which $\noise_1\in\setA_1$, $\noise_0\in\setA_0$, for
    $\alpha\in[0,1]$. The Minkowski sum preserves convexity: if $\setA_1$ and $\setA_2$ are convex, then $\setA_\alpha$ is also convex.

    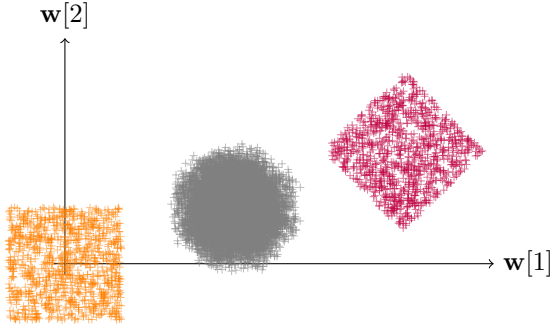
\begin{figure}
        \centering
\begin{tikzpicture}[scale=1.5]
    \draw [->] (-0.1,0) -- (3.8,0) node [anchor = west] {$\noise[1]$};
    \draw [->] (0,-0.1) -- (0,2) node [anchor = south] {$\noise[2]$};
    \draw[orange,opacity=0.5] plot[only marks,mark=+,mark size = 1] file{figs/square_minkowski.txt};
    \draw[purple,opacity=0.5] plot[only marks,mark=+,mark size = 1] file{figs/rotated_minkowski.txt};
    \draw[gray,opacity=0.5] plot[only marks,mark=+,mark size = 1] file{figs/minkowski_average.txt};
\end{tikzpicture}
        \caption{Example of the
            Minkowski sum of two sets scaled by $\alpha=1/2$ (in the center).
            Each set is represented here by random elements within it.
            \label{fig:minkowski}}
    \end{figure}

    To prove Theorem~\ref{th:result}, we identify the sets
    $\ToPars{z}{\locpar_\alpha}{\scapar_\alpha}$ in \eqref{eq:topars}
    with the sets $\setA_\alpha$ in Theorem~\ref{th:prekopa}. Then,
    \eqref{eq:prekopa} becomes the desired logconcavity statement.
    The technical conditions of Theorem~\ref{th:result}, therefore, only ensure that the convex combination of location and scale
    parameters leads to the same set $\ToPars{z}{\locpar_\alpha}{\scapar_\alpha}$ as the Minkowski sum
    of the corresponding scaled sets $\ToPars{z}{\locpar_i}{\scapar_i}$ for $i\in\lbrace 0,1\rbrace$.
    We start by verifying convexity in the extreme cases $\alpha\in\lbrace 0,1\rbrace$.

    \begin{lemma}[Convex Quantizers Lead to Convex Noise Regions]\label{lem:convex}
        For any $z\in\QuanSet$, $\locpar\in\reals^m$, and 
        $\scapar\in\PDreals{n}$, 
        $\ToPars{z}{\locpar}{\scapar}$ is convex
        if and only if $\invquantizer(z)$ is convex.
    \end{lemma}
    \begin{proof}
        Let $\noise_0,\noise_1\in
        \ToPars{z}{\locpar}{\scapar}$. Then,
        $\noise_i = \scapar \contdata_i -
        \obsmat \locpar$ for $\contdata_i \in
        \invquantizer(z)$ with $i\in\lbrace 0,1\rbrace$.
        Because $\invquantizer(z)$ is convex,
        $\contdata_\alpha\in\invquantizer(z)$, and therefore,
        $\noise_\alpha =
        \scapar \contdata_\alpha - \obsmat \locpar_\alpha \in \ToPars{z}{\locpar}{\scapar}$.
        In conclusion, if $\invquantizer(z)$ is convex,
        $\ToPars{z}{\locpar}{\scapar}$ is convex.

        For the converse, simply consider that
        $\ToPars{z}{\zero{}}{\id{n}} = \invquantizer(z)$.
    \end{proof}
    
    This allows us to set 
    \begin{equation}  \label{eq:extremes}
        \setA_i \coloneqq \ToPars{z}{\locpar_i}{\scapar_i}\mbox{ for }i\in\lbrace0,1\rbrace,
    \end{equation}
    while fulfilling the conditions of Theorem~\ref{th:prekopa}.
    For the intermediate values $\alpha\in(0,1)$, we identify conditions under 
    which $\ToPars{z}{\locpar_\alpha}{\scapar_\alpha}$ is the Minkowski sum of
    $\alpha \setA_1$ and $(1-\alpha)\setA_0$. We start by showing that one of 
    the inclusions is always true.

    \begin{lemma}
    \label{lem:combination_parameters_in_Minkowski}
        Consider $\setA_0$ and $\setA_1$ as defined in \eqref{eq:extremes}. Then
        $\ToPars{z}{\locpar_\alpha}{\scapar_\alpha} \subseteq \setA_\alpha$.
    \end{lemma}
    \begin{proof}
        Let $\noise \in\ToPars{z}{\locpar_\alpha}{\scapar_\alpha}$. Then, there is
        $\contdata\in\invquantizer(z)$ such that 
        $\noise = \scapar_\alpha \contdata - \obsmat\locpar_\alpha = \alpha\noise_1+ (1-\alpha) \noise_0$
        with $\noise_i = \scapar_i \contdata - \obsmat \locpar_i$ for $i\in\lbrace 0,1\rbrace$.
        By definition, $\noise_i\in\setA_i$. 
    \end{proof}

    To establish that $\setA_\alpha \subseteq \ToPars{z}{\locpar_\alpha}{\scapar_\alpha}$, we need results on the geometry of sets generated by normalized matrices that sum to $\id{}$ (see Lemma~\ref{lem:diag_box} in the Appendix).
    This is not true for generic scale parameters (see Lemma~\ref{lem:pd_ball}), but it holds under the restrictions of Theorem~\ref{th:result}.

    \begin{lemma}
    \label{lem:Minkowski_in_convex_parameters}
        Consider $\setA_0$ and $\setA_1$ as defined in \eqref{eq:extremes}. If, 
        \renewcommand{\theenumi}{\alph{enumi}}
        \begin{enumerate}
            \item $\scapar_1 = \scapar_0$, or,
            \item $\scapar_i = \scalarscapar_i \id{n}$ with $\scalarscapar_i>0$ for $i\in\lbrace 0,1\rbrace$, or,
            \item $\scapar_i = \diagscapar_i$ with $\diagscapar_i\in \diagPD{n}$ for $i\in\lbrace 0,1\rbrace$, with $\invquantizer(z)=\prod_{j=1}^n [a_j(z),b_j(z)]$ for $a_j(z),b_j(z)\in\overline{\reals}$ and for all $j\in\lbrace 1,2,\dots,n\rbrace$ and any $z\in\QuanSet$,
        \end{enumerate}
        \renewcommand{\theenumi}{enumi}
        then $\setA_\alpha \subseteq \ToPars{z}{\locpar_\alpha}{\scapar_\alpha}$.
        Here, $\diagPD{n}$ is the space of $n\times n$ non-negative diagonal matrices.
    \end{lemma}
    \begin{proof}
        Let $\alpha_0=1-\alpha$ and $\alpha_1=\alpha$ and consider the matrices
        \begin{IEEEeqnarray*}{c}
            \matrixC_i = (\alpha_0\scapar_0 + \alpha_1\scapar_1)^{-1} \alpha_i \scapar_i
        \end{IEEEeqnarray*}
        for $i\in\lbrace 0,1\rbrace$. Consider also that $\matrixC_0+\matrixC_1=\id{n}$.

        Let $\noise\in\setA_\alpha$. Then, there are
        $\noise_i\in\setA_i$ for $i\in\lbrace 0,1\rbrace$ such that
        $\noise=\alpha_0\noise_0+\alpha_1\noise_1$. Furthermore, by \eqref{eq:extremes}
        we have that $\noise_i = \scapar_i \contdata_i - \obsmat\locpar_i$ for $i\in\lbrace0,1\rbrace$,
        where $\contdata_i\in\invquantizer(z)$. Therefore,
        \begin{align*}
            \noise &=
            \sum_{i=0}^1
                \left(
                    \alpha_i \scapar_i \contdata_i -
                    \obsmat \alpha_i\locpar_i
                \right) \\
            &= \left(\alpha_0 \scapar_0 + \alpha_1 \scapar_1\right) (\matrixC_0\contdata_0 + \matrixC_1\contdata_1) 
            - \obsmat \left(\alpha_0\locpar_0+\alpha_1\locpar_1\right).
        \end{align*}
        Then, $\noise \in \ToPars{z}{\locpar_\alpha}{\scapar_\alpha}$
        if and only if $\contdata = \matrixC_0\contdata_0 + \matrixC_1\contdata_1\in\invquantizer(z)$.

        If condition a) is fulfilled, then $\matrixC_i = \alpha_i\id{n}$ and
                $\contdata=\contdata_\alpha$. Because $\invquantizer(z)$ is convex,
                $\contdata\in\invquantizer(z)$.
        If condition b) is fulfilled, then
                $\matrixC_i = \tilde{\alpha}_i\id{n}$ with $\tilde{\alpha}_i = \alpha_i \scalarscapar_i
                / (\alpha_0\scalarscapar_0 + \alpha_1 \scalarscapar_1)$, and
                $\contdata$ is a convex combination of $\contdata_0$ and $\contdata_1$, i.e.,
                $\contdata=\contdata_{\tilde{\alpha}_1}$. Because $\invquantizer(z)$ is convex,
                $\contdata\in\invquantizer(z)$.
        If condition c) is fulfilled, then the $\matrixC_i \in\diag{n}{[0,1]}$ for $i\in\lbrace0,1\rbrace$, where $\diag{n}{[0,1]}$ is the set of $n\times n$ diagonal matrices with elements in $[0,1]$.
                By Lemma~\ref{lem:diag_box}, we then have
                $\contdata\in\prod_{j=1}^n \left[\contdata_1[j],\contdata_2[j]\right]$.
                Because $\invquantizer(z)=\prod_{j=1}^n [a_j(z),b_j(z)]$, and 
                $a_j(z)\leq \contdata_1[j],\contdata_2[j] \leq b_j(z)$,
               we have that $\contdata\in\invquantizer(z)$.
        Therefore, if either a), b) or c) are given, 
        $\noise \in \ToPars{z}{\locpar_\alpha}{\scapar_\alpha}$.
    \end{proof}

    We can now combine Lemmas~\ref{lem:convex}--\ref{lem:Minkowski_in_convex_parameters} with Theorem~\ref{th:prekopa} to show Theorem~\ref{th:result}.

    \begin{proof}[Proof of Theorem~\ref{th:result}]
				Consider $\setA_0$ and $\setA_1$ as defined in \eqref{eq:extremes}.
				By Lemma~\ref{lem:convex}, they are convex sets.
				By Lemmas~\ref{lem:combination_parameters_in_Minkowski}~and~\ref{lem:Minkowski_in_convex_parameters},
				if a), b) or c) are fulfilled, we have that
				$\setA_\alpha = \ToPars{z}{\locpar_\alpha}{\scapar_\alpha}$. Then, Theorem~\ref{th:prekopa} and \eqref{eq:quantized_likelihood} yield
				\begin{IEEEeqnarray*}{rl}
					\likelihood{\locpar_\alpha,\scapar_\alpha}{z} &= \prob{\noise}\left[ \setA_\alpha \right]\\
					&\geq \prob{\noise}\left[\setA_1\right]^\alpha \prob{\noise}\left[\setA_0\right]^{(1-\alpha)}\\
					&=\likelihood{\locpar_1,\scapar_1}{z}^\alpha
		            \likelihood{\locpar_0,\scapar_0}{z}^{1-\alpha}.
				\end{IEEEeqnarray*}
    \end{proof}

    \section{Acknowledgments}

		We acknowledge access to the facilities and expertise of the CIBM Center
for Biomedical Imaging, a Swiss research center of excellence founded and
supported by Lausanne University Hospital (CHUV), University of Lausanne
(UNIL), École polytechnique fédérale de Lausanne (EPFL), University of
Geneva (UNIGE), and Geneva University Hospitals (HUG).

		This work was partially supported by the SRA ICT TNG project Privacy-preserved Internet Traffic Analytics (PITA).

	\appendix[Matrix combinations] 
            In the proof of Lemma~\ref{lem:Minkowski_in_convex_parameters}, we use that sets of the form $\prod_{j=1}^n [a_j,b_j]$ with $a_j,b_j\in\overline{\reals}$ are closed with respect to the
generalization of convex combinations to diagonal matrices. In Figure~\ref{fig:matrix_combs}a,
we include an illustration of a practical case in $2$D. Here, we provide the proof. To our knowledge, this has not been reported before.

\begin{lemma}[Diagonal Matrices Whose Sum is the Identity Generate Squares] \label{lem:diag_box}
    Let $\diag{n}{[0,1]}$ be the set of $n\times n$ diagonal matrices with elements in $[0,1]$, and let $\contdata_0,\contdata_1\in\reals^n$. Then,
    \begin{align*}
        \mathcal{H} &\coloneqq
        \left\lbrace 
            \matrixC \contdata_0 + (\id{n}-\matrixC) \contdata_1 : 
            \matrixC \in \mathcal{D}_{n}\left[0,1\right]
        \right\rbrace \\
        &= \prod_{j=1}^n \left[ \contdata_0[j],\contdata_1[j] \right] \eqqcolon \mathcal{H}_{\blacksquare}\,.
    \end{align*}    
\end{lemma}
\begin{proof}
    For $\mathcal{H}_{\blacksquare}\subseteq\mathcal{H}$, let $\contdata \in \mathcal{H}_{\blacksquare}$. 
    If $\alpha_j = (\contdata[j]-\contdata_1[j])/(\contdata_0[j]-\contdata_1[j])$, then $\alpha_j\in[0,1]$.
    If $\matrixC\in\diag{n}{[0,1]}$ is the diagonal matrix such that 
    $\matrixC[j,j]=\alpha_j$, then 
    $\matrixC \contdata_0 + (\id{n}-\matrixC) \contdata_1 = \contdata$. Thus,
    $\contdata\in\mathcal{H}$.
    
    For $\mathcal{H}\subseteq \mathcal{H}_{\blacksquare}$, let $\contdata\in\mathcal{H}$. 
    Then, we have that $\alpha_j=\matrixC[j,j]\in[0,1]$, and 
    $\contdata[j] = \alpha_j \contdata_{1}[j] + (1-\alpha_j)\contdata_0[j]$, and thus,
    $\contdata[j]\in [\contdata_0[j], \contdata_1[j]]$. Therefore, 
    $\contdata \in \mathcal{H}_{\blacksquare}$.
\end{proof}

\begin{figure}
    \centering
    \begin{tikzpicture}[scale=1.2]
     \draw [->] (-.05,0) -- (1.2,0) node[anchor = west] {$\contdata[1]$}; 
    \draw [->] (0,-.05) -- (0, 0.7) node[anchor= south] {$\contdata[2]$};
    \draw[red] (1,-.5) rectangle (-1,.5);
    \draw[black,opacity=0.5] plot[only marks,mark=*,mark size = 0.1] file{figs/rectangle.txt};
    \node[blue,opacity=0.5] (p0) at (-1,0.5) {$\bullet$};
    \node at (p0) [anchor=south east] {$\contdata_0$};
    \node[blue,opacity=0.5] (p1) at (1,-0.5) {$\bullet$};
    \node at (p1) [anchor=north west] {$\contdata_1$};
    \node at (-1.3,0) {a)};
\end{tikzpicture}
\begin{tikzpicture}[scale=1.2]
     \draw [->] (-.05,0) -- (1.3,0) node[anchor = west] {$\contdata[1]$}; 
    \draw [->] (0,-.05) -- (0, 1.3) node[anchor= south] {$\contdata[2]$};
    \draw[red] (0,0) circle (1.118);
    \draw[black,opacity=0.5] plot[only marks,mark=*,mark size = 0.1] file{figs/circle.txt};
    \node[blue,opacity=0.5] (p0) at (-1,0.5) {$\bullet$};
    \node at (p0) [anchor=south east] {$\contdata_0$};
    \node[blue,opacity=0.5] (p1) at (1,-0.5) {$\bullet$};
    \node at (p1) [anchor=north west] {$\contdata_1$};
    \node at (-1.3,-.25) {b)};
\end{tikzpicture}

\vspace{5pt}
    
    \vspace{-5pt}
    
    \caption{In black, points obtained by combinations with matrices whose sum is
        the identity matrix, i.e., $\matrixC\contdata_1 + (\id{n}-\matrixC)\contdata_0$,
        where $\matrixC$ were, in a), random diagonal matrices from $\diag{n}{[0,1]}$, and, in b),
        random positive semidefinite matrices from $\PDreals{n}$ with 
        $\rho\left(\matrixC\right)\leq 1$. In blue, $\contdata_0$ and $\contdata_1$.
    \label{fig:matrix_combs}}
\end{figure}

On the one hand, Lemma~\ref{lem:diag_box} allows for the most general result in terms of the scale parameter 
$\scapar$ we have obtained (Theorem~\ref{th:result}.c). 
On the other hand, the corresponding result for arbitrary scale parameters suggests that the strategy behind our proof of Theorem~\ref{th:result} might not 
generalize well. We include it here in Lemma~\ref{lem:pd_ball} for completeness. To our knowledge, it has not been reported before.

\begin{lemma}[Positive Semidefinite Matrices that Sum to the Identity Generate Balls\footnote{ See \url{www.geogebra.org/m/hdxtmz3b} and \url{www.geogebra.org/m/tskjev2m} for 2D and 3D dynamic examples of Lemma~\ref{lem:pd_ball} based on the SVD of any positive semidefinite matrix.}] \label{lem:pd_ball}
    Let $\PDbounded{n}$ be the set of real $n\times n$ symmetric positive-semidefinite matrices
    with spectral radius less or equal than $1$, i.e., $\rho\left(\matrixC\right)\leq 1$, and $\contdata_0,\contdata_1\in\reals^n$. Then
    \begin{align*}
        \mathcal{S} & \coloneqq \left\lbrace \matrixC\contdata_0 
        + \left(\id{n}-\matrixC\right)\contdata_1 : \matrixC \in \PDbounded{n}  \right\rbrace \\
        &= \mathcal{B}\left(\frac{\contdata_0+\contdata_1}{2},\frac{1}{2} \left\Vert \contdata_1 - \contdata_0 \right\Vert_2\right)\eqqcolon \mathcal{S}_\bullet\,,
    \end{align*}
    where $\mathcal{B}(\contdata_\mathrm{c},r)$ is the closed $\ell_2$ ball centered at $\contdata_\mathrm{c}\in\reals^n$ with radius $r>0$.
\end{lemma}
\begin{proof}
    For $\mathcal{S}_\bullet\subseteq\mathcal{S}$, let
    $\contdata\in\mathcal{S}$. Then, there is a $\matrixC \in \PDbounded{n}$ such that 
    $\contdata = \matrixC \contdata_0 + \left(\id{n}-\matrixC\right)\contdata_1$.
    If $\contdata_\mathrm{c} = (\contdata_0+\contdata_1)/2$, then
    \begin{IEEEeqnarray*}{rl}
			\left\| \contdata - \contdata_\mathrm{c} \right\|_2\, & = \left\| \left( \matrixC - \frac{\id{n}}{2}\right) 
			\left(\contdata_0 - \contdata_1 \right) \right\|_2 \\
			&\leq \biggl\|\matrixC - \frac{\id{n}}{2}\biggr\|_{2\to 2} \left\|\contdata_0 - \contdata_1 \right\|_2 \\
			&\leq \frac{1}{2} \left\| \contdata_1 - \contdata_0 \right\|_2\,.
    \end{IEEEeqnarray*}
    Therefore, $\contdata \in \mathcal{S}_\bullet$.
    Here, $\Vert\cdot\Vert_{2\to 2}$ is a the operator norm of a matrix with respect to the $\ell_2$ norm, and we have used that it coincides with the 
    spectral radius $\rho(\cdot)$ for Hermitian matrices.
    
    For $\mathcal{S}_\bullet\subseteq\mathcal{S}$, let $\contdata\in\mathcal{S}_\bullet$ and
    $\contdata_\mathrm{c} = (\contdata_0+\contdata_1)/2$. Then, consider $\tilde{\contdata} = 
    \contdata-\contdata_1$ and $\tilde{\contdata}_0=\contdata_0-\contdata_1$. Because
    $\contdata\in\mathcal{S}_\bullet$, we have that 
    $\left\| \tilde{\contdata} - \tilde{\contdata}_0/2\right\|_2^2 = 
    \left\| \contdata - \contdata_\mathrm{c}\right\|_2^2 
    \leq \left(\left\| \contdata_1 - \contdata_0 \right\|_2/2\right)^2 = 
    \left(\left\|\tilde{\contdata}_0\right\|_2/2\right)^2$. 
    Expanding the squares, we obtain $0\leq\left\|\tilde{\contdata}\right\|_2^2 \leq \tilde{\contdata}_0^{\mathrm{T}} \tilde{\contdata}$.
    If we consider then the matrix 
    $\matrixC = \tilde{\contdata}\tilde{\contdata}^{\mathrm{T}} / 
    \tilde{\contdata}^{\mathrm{T}} \tilde{\contdata}_0$, we see that it is a 
    rank $1$ matrix with a single non-zero eigenvalue 
    $\lambda = \mathrm{tr}\left(\matrixC\right) = 
    \left\| \tilde{\contdata}\right\|^2 / \tilde{\contdata}^{\mathrm{T}} \tilde{\contdata}_0
    \in[0,1]$, i.e., $\matrixC \in \PDbounded{n}$. 
    Furthermore, $\matrixC\left(\contdata_0 -
    \contdata_1\right) = \matrixC \tilde{\contdata}_0 = \tilde{\contdata} = \contdata - 
    \contdata_1$, i.e., $\contdata = \matrixC \contdata_0 + (\id{n}-\matrixC)\contdata_1$. 
    Therefore, $\contdata \in \mathcal{S}$.

\end{proof}

\begin{remark}
    To apply our proof technique for Lemma~\ref{lem:Minkowski_in_convex_parameters} to arbitrary scale parameters $\scapar\in\PDreals{n}$, we would need to restrict the quantization regions so that $\contdata_0,\contdata_1 \in \invquantizer(z)$ implies $\mathcal{S}_\bullet \subset \invquantizer(z)$. Our intuitive understanding is that only trivial quantizers would fulfill this property.    However, because Theorem~\ref{th:prekopa} is sufficient but not necessary, Lemma~\ref{lem:pd_ball} does not preclude 
    likelihood logconcavity results for that case.
\end{remark}

	\bibliographystyle{IEEEtranMod}
	\bibliography{QML}

\end{document}